\documentclass[final]{siamart1116}

\usepackage{graphicx}
\usepackage{subfigure}
\usepackage{epstopdf}
\usepackage{amsmath}
\usepackage{amsfonts}

\usepackage{algorithm}
\usepackage{algpseudocode}

\title{A Quantum Walk Enhanced Grover Search Algorithm for Global Optimization
} 

\author{Yan Wang
\thanks{Georgia Institute of Technology, Atlanta, Georgia. 
(\email{yan.wang@me.gatech.edu})}}

\begin{document}
\maketitle
\slugger{mms}{xxxx}{xx}{x}{x--x}%slugger should be set to mms, siap, sicomp, sicon, sidma, sima, simax, sinum, siopt, sisc, or sirev

%\linenumbers

\begin{abstract}
One of the significant breakthroughs in quantum computation is Grover's algorithm for unsorted database search.  Recently, the applications of Grover's algorithm to solve global optimization problems have been demonstrated, where unknown optimum solutions are found by iteratively improving the threshold value for the selective phase shift operator in Grover rotation. In this paper, a hybrid approach that combines continuous-time quantum walks with Grover search is proposed so that the search is accelerated with improved threshold values. By taking advantage of the quantum tunneling effect, better threshold values can be found at the early stage of the search process so that the sharpness of probability improves. The results between the new algorithm, existing Grover search, and classical heuristic algorithms are compared. 

\end{abstract}



\pagestyle{myheadings}
\thispagestyle{plain}
\markboth{QUANTUM WALK ENHANCED GROVER OPTIMIZATION}{Y. WANG}

\section{Introduction}  \label{sec:intro}

The potential of quantum computation to solve scientific and engineering problems has been recognized in the past decade. The power of quantum computers is in both time and space efficiency. The major exciting breakthroughs include the discovery of Shor's algorithm \cite{shor1994algorithms} that factors integers in polynomial times which is exponentially faster than any of the previously known classical ones, and Grover's algorithm \cite{grover1996fast} for unsorted database search which has the quadratic speedup. In locating one out of $N$ items, Grover's search algorithm requires only $O(\sqrt{N})$ iterations of a so-called Grover rotation that consists of a selective phase shift operator and a Grover operator. In general, to locate one of $m$ solutions out of a total of $N$ possible items if $m$ is known, the upper bound for the number of Grover rotations is $\lceil(\pi/4)\sqrt{N/m}\rceil$.

Recently, the applications and extensions of Grover's algorithm to solve global optimization problems have been demonstrated. D\"{u}rr and H{\o}yer \cite{durr1996quantum} first applied Grover's algorithm in optimization by randomly selecting a possible solution, using its functional evaluation as the threshold in the selective phase shift operator, and applying a certain number of Grover rotations for each optimum search iteration. The number of Grover rotations is increased gradually based on the upper bound of Grover search with unknown number of solutions \cite{boyer1998tight}. Bulger et al. \cite{bulger2003implementing} took an adaptive search strategy to change the number of Grover rotations per iteration dynamically, where the number of Grover rotations is also randomly sampled between zero and the incremental limit. Baritompa et al. \cite{baritompa2005grover} developed a further improved adaptive algorithm where the number of Grover rotations for each iteration is determined by a strategy of maximizing the benefit-cost ratio as the expected value gain to the number of rotations. A sequence of rotation numbers was also generated to heuristically implement the strategy. Bulger \cite{bulger2007combining} combined Grover's search algorithm with local search techniques where Grover's algorithm is only used to locate the basin that possibly contains the global optimum solution.  Liu and Koehler \cite{liu2010using} provided a different strategy where Bayesian update is applied to determine the benefit-cost ratio. 
%Liu and Koehler \cite{liu2012hybrid} also further improved the computational efficiency by taking only one Grover rotation at the early stage of search when the threshold is far from the global optimum and at least one fourth of all possible solutions have better functional evaluations than the threshold, in which case the selective phase shift operator takes $\cos^{-1}(1/9)$ instead of $\pi$ as in the classical Grover search.

There is another category of quantum algorithms to solve global optimization problems, called adiabatic quantum optimization or quantum annealing \cite{farhi2001quantum,santoro2002theory,reichardt2004quantum,santoro2006optimization}. Based on the objective function, a Hamiltonian $H_{p}$ corresponding to the ground state of the quantum system is constructed. In quantum mechanics, a Hamiltonian is an operator corresponding to the total energy of a system. From an initial state with the Hamiltonian $H_{i}$, the quantum system with the linearly interpolated Hamiltonian $H(s)=(1-s)H_{i}+sH_{p}$ evolves toward the ground state with the gradual change of $s$ from $0$ to $1$, according to Schr\"{o}dinger's equation. If the evolution is slow enough, the system will reach the ground state and the solution of the optimization problem can be found. Compared to the temperature-based simulated annealing where the temperature of the system gradually reduces to search for global minimum, quantum annealing takes the advantage of quantum tunneling effect and tends to outperform simulated annealing. Yet the disadvantages of quantum annealing include the requirement of slow evolution and the possibility of being trapped in local minima.

Recently, a new approach \cite{wang2014global} that combines the Grover search with quantum walks was proposed to quickly improve the threshold functional value in Grover's algorithm. Other Grover methods for global optimization only considered the improvement of computational efficiency by optimizing the number of Grover rotations. There is yet another aspect of the search efficiency, which is the threshold functional value. The threshold is important in convergence speed because it determines the number of solutions $m$ out of a total of $N$ possibilities in the discretized solution space. That is, there are $m$ solutions of which the functional evaluations are better than the threshold value. If $m$ is large, the magnitude of amplitude and thus the probability of finding the optimum will not be `sharp', and the sampling of threshold is not effective in finding the actual optimum. By taking the advantage of quantum tunneling, this new approach introduces a quantum walk mechanism so as to increase the `sharpness' of probability distribution at the early stage of search. 

In this new approach, quantum walks replace the Grover rotations when the number of rotations is low during the iterative searching process. When the number of Grover rotations is low, the quantum measurement is based on an almost uniform distribution. Therefore, the chance that the threshold functional value is improved is low. In contrast, quantum walks can result in a probability distribution according to the objective function, where the probability of finding a better threshold value is higher. Quantum walk can be considered as a quantum version of the classical random walk, where a
stochastic system is modeled in terms of probability amplitudes instead of probabilities. In
random walk, the system's state $\mathbf{x}$ at time $t$ is described by a probability distribution
$p(\mathbf{x},t)$. The system evolves by transitions. The state distribution after a time period of
$\tau$ is $p(\mathbf{x},t+{\tau})=T(\tau)p(\mathbf{x},t)$ where $T(\tau)$ is the transition
operator. In quantum walk, the system's state is described by the complex-valued amplitude
$\psi(\mathbf{x},t)$. Its relationship with the probability is $\psi^{*}\psi=|\psi|^2=p$. The
system evolution then is modeled by the quantum walk
$\psi(\mathbf{x},t+{\tau})=U(\tau)\psi(\mathbf{x},t)$ with $U$ being a unitary and reversible operator. In quantum walks, probability is replaced by amplitude and Markovian dynamics is replaced by unitary dynamics.

Similar to random walks, there are discrete-time quantum walks and continuous-time quantum walks. The study of discrete-time quantum walks started from 1990s \cite{Meyer96,AmbainisBNVW01} in the context of quantum algorithm and computation \cite{Kempe03,Kendon07,Konno08}. Although the term, continuous-time quantum walk, was introduced more recently \cite{FarhiGutmann1998}, the research of the topic can be traced back much earlier in studying the dynamics of quantum systems, particularly in the path integral formulation of quantum mechanics generalized by Feynman \cite{Feynman1948} in 1940's. The relationship between the discrete- and continuous-time quantum walks was also studied. The two models have similar speed performance and intrinsic relationships. The convergence of discrete-time quantum walks toward continuous-time quantum walks has been demonstrated \cite{Strauch2006,Childs2010}. 

Here, we take the continuous-time quantum walk approach to introduce tunneling. This extra step helps accelerate the Grover search for global optima by increasing the sampling probabilities of the global optimum states through quantum accelerated diffusion. In this paper, this hybrid optimization approach is described in details, particularly the choice of time step in continuous-time quantum walks for the consideration of efficiency and the evaluation of objective functions on quantum computers. In the remainder of the paper, the continuous-time quantum walk formulation is first introduced in Section \ref{sec:ctqw}. The new global optimization algorithm that combines quantum walk and Grover search will be presented in Section \ref{sec:algorithm}. The computational study that simulates the quantum algorithm on the conventional computer will be described in Section \ref{sec:implementation}.

\section{Continuous-time quantum walk}  \label{sec:ctqw}

The dynamics of quantum systems is described by Schr\"{o}dinger's equation 
\begin{equation}
 i\frac{d}{dt}\psi(\mathbf{x},t)=H(t)\psi(\mathbf{x},t)
\end{equation}
where $H(t)$ is the Hamiltonian and $i=\sqrt{-1}$. Continuous-time quantum walks in one-dimensional (1-D) space can be formulated to model the quantum drift-diffusion process, described by
\begin{equation}
i\frac{\partial}{\partial{t}}\psi(x,t)=-\frac{b}{2}\frac{\partial^2}{\partial{x^2}}\psi(x,t) -iV(x,t)\psi(x,t) 
\end{equation} 
where $b$ is the diffusion coefficient and $V(x,t)$ is the potential function. Assuming that a minimization problem $\min_{x} f(x)$ is to be solved, we then have $V(x)=f(x)$. In the context of optimization, searching optima in high-dimensional solution space can be easily converted to 1-D diffusion problem with state mapping. Work has been done for high-dimensional discrete-time \cite{inui2004localization,gonulol2009decoherence} and continuous-time quantum walks \cite{wang2013simulating,wang2016accelerating}.

Path integral is a classical approach to solve the quantum dynamics problem. To construct the unitary operator $U$ that describes quantum state transitions, a general \emph{functional integral} \cite{farhi1992functional}
\begin{equation} \label{eq:F_jk}
 F_{jk}:=\int{ dq_{jk}e^{-i\int_{t_{0}}^{t_{0}+\tau}{W_{q(s)}ds}}\prod_{l \rightarrow m}e^{i\theta_{ml}} }
\end{equation}
for a path from state $x_{k}$ to state $x_{j}$ is applied. Here, $dq_{jk}$ is the probabilistic measure on the path from $x_{k}$ to $x_{j}$, which is analogous to continuous-time Markov chain model. A \emph{path} $q(s)$ is defined as a functional mapping from time $s$ to the state space. For instance, $q(t_0)=x_k$ and $q(t_0+\tau)=x_j$ represent the transitional path from state $x_{k}$ to state $x_{j}$ during a time period of $\tau$. The
overall probability of all possible paths is given as $\int_{t_{0}}^{t_0+\tau}{W_{q(s)}ds}$ from $x_{k}$ at time $t_{0}$ to $x_{j}$ at time $t_{0}+\tau$. In Eq.(\ref{eq:F_jk}), 
$e^{-i\int_{t_{0}}^{t_0+\tau}{W_{q(s)}ds}}$ can be regarded as the \emph{weight} of transition from $x_{k}$ to $x_{j}$, and
$\prod_{l \rightarrow m}e^{i\theta_{ml}}$ is the \emph{total phase shift factor} for all jumps in
transition from $x_{k}$ to $x_{j}$, where each of $e^{i\theta_{ml}}$ corresponds to the phase shift for one
of the jumps during the transition. 

Similar to the classical Chapman-Kolmogorov equation of state transitions, a transition rate from state $x_{k}$ to state $x_{j}$ at time $t$ in terms of probability amplitude is
\begin{equation}
 \rho_{jk}e^{i\theta_{jk}}:=-i\langle x_{j} | H(t) | x_{k} \rangle
\end{equation}
where $\rho_{jk}$ is the magnitude of transition rate and $\theta_{jk}$ is the phase. Then the magnitude of leaving state $x_{k}$ is
\begin{equation}
 \rho_{k}:=\sum_{k\neq j}\rho_{jk}
\end{equation}
and the overall transition rate for state $x_{k}$ is determined by
\begin{equation}
 W_{k}:= \langle x_{k} | H(t) | x_{k} \rangle + i\rho_{k}
\end{equation}

The elements of the Hamiltonian matrix $\hat{H}$ for 1-D lattice space that has integer indices and the spacing $\Delta$ are given by
\begin{equation} \label{eq:elementHamiltonDriftDiff}
 \langle j|\hat{H}|k\rangle
=-\frac{b}{2\Delta^2}\delta_{j,k-1}+(\frac{b}{\Delta^2}-iV_{k})\delta_{j,k}-\frac{b}{2\Delta^2}
\delta_{j,k+1}
\end{equation}
where $\delta_{j,k}$ is the Kronecker delta, and the states are simply denoted by integers as $x=\ldots,-2,-1,0,1,2,\ldots$. 

For a transitional path with $k \neq j$, 
\begin{equation*}
\begin{split}
 \rho_{jk}&= \frac{b}{2\Delta^2}[\delta_{j,k-1}+\delta_{j,k+1}] \\
 e^{i\theta_{jk}}&=i \\
 \rho_{k}&=\rho_{k-1,k}+\rho_{k+1,k}=\frac{b}{\Delta^2} \\
 W_{k}&=\frac{b}{\Delta^2}-iV_{k}+i\frac{b}{\Delta^2} 
\end{split}
\end{equation*}

\subsection{Functional integral}

Consider that the 1-D transitional paths are memoryless and the transition rate is
$b/(2\Delta^2)$ per unit time. The numbers of jumps to the left or right direction
within a time period follows a Poisson distribution. That is, the probability that there are $l$
jumps to the left for time $\tau$ is $e^{-b\tau/(2\Delta^2)}(b\tau/(2\Delta^2))^{l}/{l!} $.
Similarly it is $e^{-b\tau/(2\Delta^2)}(b\tau/(2\Delta^2))^{r}/{r!} $ for $r$ jumps to the
right. Assuming the final state is at $n$ steps away and on the right to the initial state, $r-l=n$.
The probabilistic measure $dq_{jk}$ in Eq.(\ref{eq:F_jk}) for one path from state $0$ to $n\Delta$ that has
$l$ left jumps is
\begin{equation}  \label{eq:dq_n0}
dq_{n,0}^{(l)}=\frac{e^{-b\tau/(2\Delta^2)}(b\tau/(2\Delta^2))^{l}}{l!}\frac{e^{-b\tau/(2\Delta^2)}
(b\tau/(2\Delta^2))^{n+l}}{(n+l)!} 
\end{equation}

For a transition with $n$ steps away from the initial state for a total period $\tau$, the weight in the functional integral can be calculated as 
\begin{equation*}
e^{-i\sum_{l}{W_{l}\tau_{l}}}=e^{-i\sum_{l}{[\frac{b}{\Delta^2}+i(\frac{b}{\Delta^2}-V_{l})]\tau_{l}}}
\approx e^{(1-i)\frac{b}{\Delta^2}\tau-V_{n}\tau}
\end{equation*}
where $V_{n}$ denotes the potential at the final state and $\sum_{l}\tau_{l}=\tau$. With the probabilistic measure as in Eq.(\ref{eq:dq_n0}), the functional integral for quantum drift-diffusion processes in Eq.(\ref{eq:F_jk}) becomes
\begin{equation}   \label{eq:F_n0_quantumDD}
%\begin{split}
 F_{n,0}=\sum_{l=0}^{\infty} [dq_{n,0}^{(l)}e^{-i\tau (1+i)b/\Delta^2-V_{n}\tau}(-1)^{l}i^{n}] 
% \\
%&= \sum_{l=0}^{\infty}
%e^{-b\tau/\Delta^2}\frac{(\frac{b\tau}{{2\Delta^2}})^{2l+n}}{l!(n+l)!} e^{-i\tau(1+i)b/\Delta^2-V_{n}\tau}(-1)^{l}i^{n} \\
%&=i^{n}e^{-ib\tau/\Delta^2-V_{n}\tau}\sum_{l=0}^{\infty}
%\frac{(-1)^{l}(\frac{b\tau}{{2\Delta^2}})^{2l+n}}{l!(n+l)!}  \\
= i^{n}e^{-ib\tau/\Delta^2-V_{n}\tau}J_{n}(\frac{b\tau}{\Delta^2})
%\end{split}
\end{equation}
where $J_{n}(y)$ is the \emph{Bessel function of first kind} with integer order $n$ and
input $y$ ($y \geq 0$). Additionally, $J_{-n}(y)=(-1)^{n}J_{n}(y)$.

Based on Eq.(\ref{eq:F_n0_quantumDD}), the elements of the unitary quantum walk operator $U=(u_{jk})_{N \times N}$ are updated as $u_{jk}=F_{(j-k),0}$ for the given space resolution $\Delta$ and time resolution $\tau$. 

\subsection{Choice of time step $\tau$}

Compared to random walk, the power of quantum walk lies in its capability of capturing the long-range spatial correlation and thus the tunneling effect. This is largely due to the Bessel function. Given the amplitudes $\psi(t)$ associated with all states at time $t$, one step of quantum walk will yield $\psi(t+\tau)$ with the $j^{th}$ element ($j=1,\ldots,N$) updated by 
\begin{equation*}
\psi_{j}(t+\tau)=\sum_{k}F_{(j-k),0}\psi_{k}(t)
\end{equation*}
Consider that the system starts at state $K$ with $\psi_{K}(t)=1.0$ and $\psi_{k\neq K}(t)=0.0$, where $K$ is any index between $1$ and $N$. The $j^{th}$ element is then updated to
\begin{equation*}
\psi_{j}(t+\tau)=F_{(j-K),0}=i^{(j-K)}e^{-ib\tau/\Delta^2-V_{j}\tau}J_{(j-K)}(\frac{b\tau}{\Delta^2})
\end{equation*}
The corresponding updated probability that state $j$ is observed is
\begin{equation}   \label{eq:probStateUpdate}
\Pr(x=j)=\psi_{j}^{*}(t+\tau)\psi_{j}(t+\tau)=e^{-2V_{j}\tau}J_{(j-K)}^{2}(\frac{b\tau}{\Delta^2})
\end{equation}
With Eq.(\ref{eq:probStateUpdate}), the probability of arriving certain state can be adjusted by selecting appropriate time step $\tau$ and diffusion coefficient $b$. Given that the Bessel functions of the first kind $J_{n}$'s are continuous and oscillatory with values between $-1$ and $1$, $\Pr(x=j)$ in Eq.(\ref{eq:probStateUpdate}) has the local maximum values where $\tau$ satisfies $\partial{J_{(j-K)}}/\partial{\tau}=0$ with fixed $b$ and $\Delta$. 

The first derivatives of Bessel function $J_{n}(z)$'s with respect to $z$ can be derived and calculated recursively, as 
%From the recurrence forms of Bessel function
%\begin{equation*}
%\frac{d}{dz}[z^{n}J_{n}(z)]=z^{n}J_{n-1}(z)
%\end{equation*}
%and
%\begin{equation*}
%\frac{d}{dz}[z^{-n}J_{n}(z)]=-z^{-n}J_{n+1}(z)
%\end{equation*}
%we can further receive
%\begin{equation}  \label{eq:recurrenceJn1}
%nz^{-1}J_{n}(z)+J_{n}'(z)=J_{n-1}(z)
%\end{equation}
%and
%\begin{equation}  \label{eq:recurrenceJn2}
%-nz^{-1}J_{n}(z)+J_{n}'(z)=-J_{n+1}(z)
%\end{equation}
%respectively. The summation of Eqs.(\ref{eq:recurrenceJn1}) and (\ref{eq:recurrenceJn2}) yields
\begin{equation}  \label{eq:zero_deriv_Jn}
J_{n}'(z)=\frac{1}{2}[J_{n-1}(z)-J_{n+1}(z)]
\end{equation}
%Additionally, from Eq.(\ref{eq:recurrenceJn2}) 
%\begin{equation*}
with $J_{0}'(z)=-J_{1}(z)$.
%\end{equation*}

The zeros of $J_{n}'(z)$'s determine where the local maximum probabilities in Eq.(\ref{eq:probStateUpdate}) are obtained with the oscillatory pattern, as shown in Figure \ref{fig:BesselFunc}. Notice that the coefficient $e^{-2V_{j}\tau}$ in Eq.(\ref{eq:probStateUpdate}) adds the modulation effect of the potential or objective function onto the probability distribution. As a result, the states with lower energy levels tend to have higher probability values. In addition, the zeros of $J_{n}(z)$'s determine where the probabilities become zeros and no samples will be drawn from those states. 

The advantage of continuous-time quantum walks is that the choice of time step $\tau$ can be arbitrarily long, and there is no need to keep it short. The choice of time step $\tau$ affects the probability value in Eq.(\ref{eq:probStateUpdate}). The temptation is to choose $\tau$ to be as large as possible such that the quantum walk can span over the major portion of state space. However, a balanced approach should be taken, because $\sum_{n}J_{n}^2(z)=1$ regardless of $z$. That is, the overall oscillatory amplitudes are reduced if large $z$'s are taken, and the advantage of introducing quantum walk over the uniform initial sampling in other Grover approaches \cite{baritompa2005grover,liu2010using} may diminish. If the global optimum solution is known a priori within some particular regions in the solution space, the time step size can be tailored so that the quantum walk do not under- or over-shoot and the regions are fully covered and without overestimation. 

Eq.(\ref{eq:zero_deriv_Jn}) also indicates that the maximum probability for a spatial walking step size $n$ is achieved at the time step where the probabilities for the spatial step size $n-1$ and $n+1$ are the same, if the effect of potential is not considered. This gives a unique pattern of spatial-temporal relationship for quantum walks. It is seen in Eq.(\ref{eq:probStateUpdate}) that the PDF is quadratically more sensitive to spatial resolution than to temporal resolution because of $J_{n}(b\tau/\Delta^2)$. A variation in $\Delta$ results in a more prominent change of PDF than a variation in $\tau$.

%
%\begin{table}
%\caption{Example solutions of $J_{n}'(z)=0$}
%\label{tab:Jnp_zeros}     
%\begin{tabular}{llllll}
%\hline\hline\noalign{\smallskip}
%n=0 & 3.83170597 & 7.01558667 & 10.17346814 & 13.32369194 & 16.47063005 \\
%n=1 & 1.84118378 & 5.33144277 & 8.53631637 & 11.7060049 & 14.86358863 \\
%n=2 & 3.05423693 & 6.70613319 & 9.96946782 & 13.17037086 & 16.34752232 \\
%n=3 & 4.20118894 & 8.0152366 & 11.34592431 & 14.58584829 & 17.78874787 \\
%\hline\hline\noalign{\smallskip}
%\end{tabular}
%\end{table}

\begin{figure}
\includegraphics[width=0.65\textwidth]{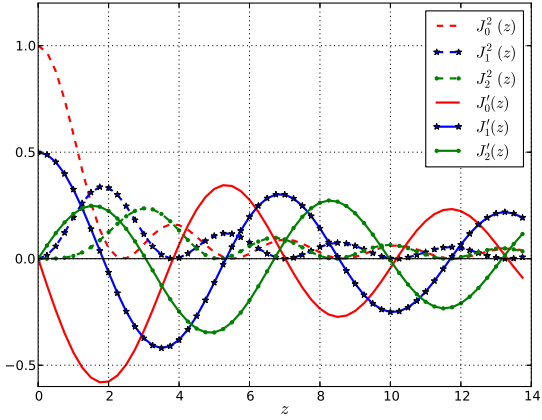} 
\caption{The first derivative of Bessel functions $J_n'$, in comparison with $J_n^2$}
\label{fig:BesselFunc}   
\end{figure}

To solve $\min f(x)$ by the optimization methods based on Grover search, if there are $m$ solutions out of a total of $N$ possible ones such that $f$ is less than a threshold value $c$, then the probability of finding a better functional evaluation after $r$ Grover rotations is $\Pr(f_{r}(x)<c)=\sin^{2}[(2r+1)\arcsin{\sqrt{m/N}}]$ \cite{boyer1998tight}. 
The sampling efficiency of quantum walk with a chosen $\tau$ is stated as follows.

\begin{theorem}
If there is a $J_{(L)}({b\tau}/{\Delta^2})$ such that $J_{(L)}^2({b\tau}/{\Delta^2}) \leq J_{(j)}^2({b\tau}/{\Delta^2})$ ($ \forall j \in \{j|V_{j}<c\}$) with diffusion coefficient $b$, spatial resolution $\Delta$, and time step $\tau$, and
%\begin{equation*}
$ \sqrt{m}e^{-c\tau}|J_{(L)}({b\tau}/{\Delta^2})| > \sin[(2r+1)\arcsin\sqrt{m/N}] $, 
%\end{equation*} 
quantum walk search is more efficient than the Grover search with the probability of finding $m$ solutions out of $N$ possible ones with the threshold value $c$ based on $r$ rotations.
\end{theorem}

\begin{proof}
From Eq.(\ref{eq:probStateUpdate}), the probability of finding a better evaluation than $c$ after one iteration of quantum walk is $\sum_{V_{j}<c}e^{-2V_{j}\tau}J_{(j)}^{2}({b\tau}/{\Delta^2})$. In order to ensure that quantum walk can locate a better threshold value, we need $\sum_{V_{j}<c}e^{-2V_{j}\tau}J_{(j)}^{2}({b\tau}/{\Delta^2}) > \sin^{2}[(2r+1)\arcsin{\sqrt{m/N}}]$. Given that $\sum_{V_{j}<c}e^{-2V_{j}\tau} \geq me^{-2c\tau}$ and if we can find a $J_{(L)}({b\tau}/{\Delta^2})$ such that $J_{(L)}^2({b\tau}/{\Delta^2}) \leq J_{(j)}^2({b\tau}/{\Delta^2})$ for all $j$ such that $V_{j}<c$, then 
\begin{equation*}
\sqrt{\sum_{V_{j}<c}e^{-2V_{j}\tau}J_{(j)}^{2}({b\tau}/{\Delta^2})} \geq \sqrt{m}e^{-c\tau}|J_{(L)}({b\tau}/{\Delta^2})|
\end{equation*}
\end{proof}

\section{The new global optimization algorithm}  \label{sec:algorithm}

The proposed algorithm starts with one iteration of continuous-time quantum walk so that the probabilities of states are distributed according to the objective function, where the minimum solutions have higher sampling probabilities during quantum measurement. 

The state-of-the-art Grover optimization algorithm is the heuristic Grover optimization algorithm \cite{baritompa2005grover,liu2010using}, where the number of Grover rotations in the iterations takes the predetermined and static sequence 
\begin{align}
Rc=(0,0, 0, 0, 1, 1, 0, 1, 1, 2, 1, 2, 3, 1, 4, 5, 1, 6, 2, 7, 9, 11, 13, 16,5, 20,  \notag \\
24, 28, 34, 2, 41,49, 4, 60, 72, 9, 88, 105, 125, 3, 149, 22, 183, 219)   
 \notag
\end{align}
That is, there is no Grover rotation in the first four iterations of search. At the beginning of each iteration in Grover search, a \emph{Hadamard} operation is applied to generate a uniform distribution in the state space. Therefore, a uniform sampling is taken in each of the first four iterations to decide the functional threshold value $c$, one rotation is taken in the fifth iteration, and so on. 

Our experiments also showed the above heuristic Grover optimization algorithm performs well. However, it is seen that the measurements at the early search process are dependent on almost uniform samplings. Our proposed algorithm replaces these small number of Grover rotations with quantum walks. That is, based on a predetermined \emph{rotation threshold value} $R_0$, if the rotation number is not greater than $R_0$, we use one step of quantum walk instead of Grover rotations. After a number of Grover rotations, a new sample is drawn from the resulting amplitude. If the functional evaluation has improved, then the new value will be used as the updated threshold $c$ for the next iteration of the Grover search. The iteration continues until certain stop criteria are met. 

%The new algorithm is listed in Table \ref{tab:qwGroverAlgorithm}. 

%In the proposed algorithm, the number of Grover rotations in each search iteration is selected to be larger than the other Grover search methods as the state of the art, denoted as the BBW-LK method \cite{baritompa2005grover,liu2010using,liu2012hybrid}. Notice that the number of rotation is proportional to $\sqrt{N/m}$ where $m$ is the number of solutions. Improved threshold objective values by quantum walks can help reduce $m$.  Therefore the number of rotations can increase accordingly. The difference will be shown in the examples in Section \ref{sec:implementation}. 

The new algorithm is listed as Algorithm \ref{tab:qwGroverAlgorithm}. The position of the initial state $x_{0}$ can be either randomly or deterministically selected with its amplitude as one.  As the search starts, one step of quantum walk is performed. The first measurement is obtained by sampling from the resulted distribution and the value is set to be the threshold $c$. During the iteration, the number of Grover rotations is based on the static sequence, except for those iterations where the number of rotations is less than $R_0$, in which case one step of quantum walk is taken instead of Grover rotations.

\begin{algorithm}
\caption{The quantum walk enhanced Grover search algorithm for minimization problems}
\label{tab:qwGroverAlgorithm}
\begin{center}
%\textbf{QW\_Grover\_Minimization\_Search()}  \\
%\emph{Input}: objective function $V(x)$, diffusion coefficient $b(t)$, \\
%        time step $\tau$, simulation time $T$  \\
%\emph{Output}: optimum solution $x^{*}$, optimum value $c$ \\
%\hline
\begin{algorithmic}[1]
\State $Rc=[0,0, 0, 0, 1, 1, 0, 1, 1, 2, 1, 2, 3, 1, 4, 5, 1, 6, 2, 7, 9, 11,13, 16, 5, 20,24, 28, 34, 2, 41,$
\State $\;\;\;\;\;\;\;\;\;\; 49, 4, 60, 72, 9, 88, 105, 125, 3, 149, 22, 183, 219]$
\State choose $\tau$, $R_0$, $b(t)$; 
\State calculate $\Delta$ based on the number of qubits and search domain; 
\State $t \gets 0$; $i \gets 0$;
\State initialize $\psi(x_{0})=1.0$ at a selected position $x_{0}$; 
\State Compute $U=F(\tau,\Delta,b(t),V)$ by Eq.(\ref{eq:F_n0_quantumDD}); 
\State $|\psi \rangle=U|\psi \rangle$;
\Comment{perform one iteration of quantum walk to find initial solution $x^{*}$} 
\State randomly sample an $x^{*}$ based on probability distribution $\psi^{2}(x)$ as quantum measurement; 
\State initialize threshold value $c=V(x^{*})$; 
\While {$i<$MAX-ITER and stop criteria not met} 
\Comment{main iterations of search}
\State $R=Rc[i]$; 
\State $i=i+1$;
\If {$R \leq R_0$}
\State initialize $\psi(x_{i})=1.0$ at a selected position $x_{i}$; 
\State Compute $U=F(\tau,\Delta,b(t),V)$ by Eq.(\ref{eq:F_n0_quantumDD}); 
\State $|\psi \rangle=U|\psi \rangle$;
\Comment{perform one iteration of quantum walk} 
\Else 
\State initialize $\psi(x)$ as a uniform distribution by the Hadamard transform; 
\For {$r=1$ to $R$}
\Comment{perform $R$ steps of Grover rotations}
\State apply Grover rotation operator to $\psi(x)$; 
\EndFor 
\EndIf
\State randomly sample an $x_{0}$ based on probability distribution $\psi^{2}(x)$ as quantum measurement; 
\If {$V(x_{0})<V(x^{*})$}  
\State $c \gets V(x_{0})$;
\Comment{update the threshold} 
\State $x^{*} \gets x_{0}$; 
%\Else  
%\State $count=count+1$; 
%\Comment{count the number of no improvement} 
\EndIf
%\If {$count>NoImprovementLowLimit$}
%\State $R=(\pi/4)\sqrt{N}$; 
%\Comment{increase the number of rotations}
%\EndIf
%\If {$count>NoImprovementHighLimit$}
%\State update $U=F(\tau,\Delta,b(t),V)$ by Eq.(\ref{eq:F_n0_quantumDD}); 
%\State $|\psi \rangle=U|\psi \rangle$;
%\Comment{perform quantum walk} 
%\State randomly sample an $x_{0}$ based on probability distribution $\psi^{2}(x)$; 
%\If {$V(x_{0})<V(x^{*})$}
%\State  $c=V(x_{0})$;
%\State  $x^{*}=x_{0}$;
%\EndIf
%\State $count=0$; 
%\EndIf
\State $t=t+\tau$; 
\EndWhile 

\end{algorithmic}
\end{center}
\end{algorithm}

\subsection{Sampling efficiency of the new algorithm}

Here we would like to show that the sampling of functional value in threshold update is more efficient in the new algorithm than the unform sampling in classical Grover optimization algorithm.

\begin{lemma}  \label{lem:1}
In solving $\min_{x\in\Omega} f(x)$, the solution sampled based on the amplitude $\psi_{QW}(x)$ resulted from quantum walk operator has better expected value than the one based on the amplitude $\psi_{H}(x)$ resulted from Hadamard operator, i.e. $\mathbb{E}_{\psi_{QW}^2}[f]<\mathbb{E}_{\psi_{H}^2}[f]$.
\end{lemma}

\begin{proof}
Suppose $|\Omega|=N$ in a discretized space. From Eq.(\ref{eq:probStateUpdate}), the expected value after one iteration of quantum walk is 
$$
\mathbb{E}_{\psi_{QW}^2}[f]= \sum_{j=1}^{N}f(j)e^{-2f(j)\tau}J_{(j-K)}^2(\frac{b\tau}{\Delta^2})/\sum_{j=1}^{N}e^{-2f(j)\tau}J_{(j-K)}^2(\frac{b\tau}{\Delta^2})
$$
for a $N$-qubit system. With $e^{-2f(j)\tau}$ in the probability density, solutions with smaller $f(x)$'s correspond to higher probabilities. Therefore, the expected value $\mathbb{E}_{\psi_{QW}^2}[f]$ is smaller than  $\mathbb{E}_{\psi_{H}^2}[f]=\sum_{j=1}^{N}f(j)/N$ based on the amplitude $\psi_{H}(x)=1/\sqrt{N}$ after Hadamard operation. 
\end{proof}

\begin{lemma}  \label{lem:2}
In solving $\min_{x\in\Omega} f(x)$,  the solution sampled based on the amplitude $\psi_{QW}(x)$ resulted from quantum walk operator has a smaller variance than the one based on the amplitude $\psi_{H}(x)$ resulted from Hadamard operator, i.e. $\mathbb{E}_{\psi_{QW}^2}[(f-\mathbb{E}_{\psi_{QW}^2}[f])^2] < \mathbb{E}_{\psi_{H}^2}[(f-\mathbb{E}_{\psi_{H}^2}[f])^2]$.
\end{lemma}

\begin{proof}
From Lemma \ref{lem:1}, we know $\mathbb{E}_{\psi_{QW}^2}[f]<\mathbb{E}_{\psi_{H}^2}[f]$. The state or solution space is divided into the following four subspaces: $\Omega_1=\{x| f \leq \mathbb{E}_{\psi_{QW}^2}[f]\}$,  $\Omega_2=\{x|\mathbb{E}_{\psi_{QW}^2}[f] < f \leq \mathbb{E}_{\psi_{H}^2}[f], |f-\mathbb{E}_{\psi_{QW}^2}[f]| \leq |f-\mathbb{E}_{\psi_{H}^2}[f]|\}$,  $\Omega_3= \{x|\mathbb{E}_{\psi_{QW}^2}[f] < f \leq \mathbb{E}_{\psi_{H}^2}[f], |f-\mathbb{E}_{\psi_{QW}^2}[f]| > |f-\mathbb{E}_{\psi_{H}^2}[f]|\}$, and $\Omega_4=\{x| f > \mathbb{E}_{\psi_{H}^2}[f]\}$. For subspaces $\Omega_1$ and $\Omega_2$, $(f-\mathbb{E}_{\psi_{QW}^2}[f])^2 \leq (f-\mathbb{E}_{\psi_{H}^2}[f])^2$. If we use $\mathbb{E}_{\psi_{QW}^2}^{(k)}[f]$ to denote the expected value for the $k$-th subspace where the probability values are the same as the original ones within the subspace and are zeros outside the subspace, then $\mathbb{E}_{\psi_{QW}^2}^{(1 \cup 2)}[(f-\mathbb{E}_{\psi_{QW}^2}[f])^2] < \mathbb{E}_{\psi_{H}^2}^{(1 \cup 2)}[(f-\mathbb{E}_{\psi_{QW}^2}[f])^2] \leq \mathbb{E}_{\psi_{H}^2}^{(1 \cup 2)}[(f-\mathbb{E}_{\psi_{H}^2}[f])^2]$.  For subspaces $\Omega_3$ and $\Omega_4$, $(f-\mathbb{E}_{\psi_{QW}^2}[f])^2 > (f-\mathbb{E}_{\psi_{H}^2}[f])^2$. Then $\mathbb{E}_{\psi_{H}^2}^{(3 \cup 4)}[(f-\mathbb{E}_{\psi_{QW}^2}[f])^2] > \mathbb{E}_{\psi_{QW}^2}^{(3 \cup 4)}[(f-\mathbb{E}_{\psi_{QW}^2}[f])^2] \geq \mathbb{E}_{\psi_{QW}^2}^{(3 \cup 4)}[(f-\mathbb{E}_{\psi_{QW}^2}[f])^2]$. The original expectations are $\mathbb{E}_{\psi_{H}^2} = \mathbb{E}_{\psi_{H}^2}^{(1 \cup 2)} + \mathbb{E}_{\psi_{H}^2}^{(3 \cup 4)}$ and $\mathbb{E}_{\psi_{QW}^2} = \mathbb{E}_{\psi_{QW}^2}^{(1 \cup 2)} + \mathbb{E}_{\psi_{QW}^2}^{(3 \cup 4)}$. Therefore, $\mathbb{E}_{\psi_{QW}^2}[(f-\mathbb{E}_{\psi_{QW}^2}[f])^2] < \mathbb{E}_{\psi_{H}^2}[(f-\mathbb{E}_{\psi_{H}^2}[f])^2]$.
\end{proof}

\begin{theorem}
In searching the global optimum in the solution space $\Omega$, sampling based on $\psi_{QW}^2(x)$ after one iteration of quantum walk provides a better solution than the one based on $\psi_{H}^2(x)$ after Hadamard operation.
\end{theorem}

\begin{proof}
For minimization problems, Lemma \ref{lem:1} shows that the expected value from samplings based on $\psi_{QW}^2(x)$ is less than the one based on $\psi_{H}^2(x)$. Lemma \ref{lem:2} shows that the variance of samples based on $\psi_{QW}^2(x)$ is also less than the one based on $\psi_{H}^2(x)$. The similar proof can be obtained for maximization problems.
\end{proof}

\begin{theorem}
In searching the global optimum in the solution space $\Omega$, sampling based on $\psi_{QW}^2(x)$ after one iteration of quantum walk results in a threshold value that is better than the one based on $\psi_{G}^2(x)$ after one iteration of Grover operation, if the quantum walk successfully locates the basin of global optimum.
\end{theorem}

\begin{proof}
It is well known that one iteration of Grover operation can locate one of $m$ solutions with the probability of one if $m=N/4$ where $N$ is the size of discrete solution space \cite{boyer1998tight}. In this case, for the minimization problem $\min_{x\in\Omega} f(x)$, $|\{x|f(x)\leq c\}|=|\Omega|/4$. One iteration of Grover operation results in $\psi_{G}^2(x)=4/|\Omega|$ for subspace $\{x|f(x)\leq c\}$ and $\psi_{G}^2(x)=0$ for subspace $\{x|f(x)>c\}$. 
\end{proof}

\subsection{Evaluation of objective functions on quantum computer}

In the original Grover search problem, a function $f$ is evaluated for an input $x$ with a binary output, either $f(x)=0$ or $f(x)=1$, as a quantum oracle or black box where $1$ indicates that $x$ is a solution and $0$ otherwise. For optimization, the complexity of the functional evaluation depends on the number of qubits. For instance, the objective function can be evaluated as $N$ of such black boxes $f_{1}(x), \ldots, f_{N}(x)$ if $N$ qubits are available in the computer. Each one of these boxes is implemented such that its output is the qubit corresponding to the evaluation. The concatenation of them represents the actual value of the objective function for the input. As a simple illustration, the evaluation of $e^{-x}$ on a three-qubit machine with input $0\le x<1$ can be implemented with three black boxes with input coded as $x=0.x_1x_2x_3$ and output $0.f_1f_2f_3$, where $f_1(x_1,x_2,x_3)=1-x_1\wedge x_2\wedge x_3$, $f_2(x_1,x_2,x_3)=[(1-x_1)\wedge(1-x_2)] \vee [(1-x_1)\wedge x_2 \wedge (1-x_3)] \vee [x_1\wedge x_2\wedge x_3]$, and $f_3(x_1,x_2,x_3)=[(1-x_1)\wedge((1-x_2)\vee(x_2 \wedge x_3))] \vee [x_1\wedge(1-x_2)\wedge(1-x_3)] \vee [x_1 \wedge x_2 \wedge x_3]$. The evaluation can be done coherently. 

\section{Implementation and numerical experiments}  \label{sec:implementation}

Both the new algorithm denoted by BBW-QW and the BBW algorithm \cite{baritompa2005grover,liu2010using} are implemented in a quantum computer emulator written in python. Experiments are conducted by several test functions, including Rastrigin $f(x)=10+x^{2}-10\cos(2\pi x)$, Schwefel $f(x)=-4.189829+30x\sin(\sqrt{|30x|})$, and Ackley $f(x)=-20\exp(-0.2|4x|)-\exp(\cos(2\pi x))$. All of the functions are challenging for local search because of the multiple steep wells of local optimums and have been widely used as benchmarks in global optimization. 

For our experiments, 9 qubits are taken to represent the discrete solution space. The corresponding spatial resolution is $\Delta=(x_{U}-x_{L})/2^9$ given the lower bound $x_L$ and upper bound $x_U$ of the considered solution range. We also implemented the BBW-LK algorithm \cite{baritompa2005grover,liu2010using}, with both dynamic and static rotation strategies. In the dynamic strategy, the numbers of rotations are calculated at run time to maximize the benefit-cost ratio for each iteration, whereas in the static strategy, the numbers of rotations are fixed as a sequence of values. Our experiments showed that the static rotation strategy actually performs more robust with higher probabilities of success for these benchmark functions used in this paper. Therefore the static strategy is used to compare with the proposed quantum walk based method. 

As shown in Figure \ref{fig:Rastrigin_avg}, the average PDF's for all possible values of $x$ in the solution space over 20 runs of search are compared between the proposed quantum walk grover search algorithm and the BBW algorithm, where the rotation threshold $R_0$ is $2$.  The typical PDF's for only one run of search by the two algorithms are compared in Figure \ref{fig:Rastrigin}. It is seen that the PDF's are flat and close to the uniform distribution for few rotations in the BBW algorithm. In the proposed BBW-QW algorithm, they are replaced by a sharper distribution after one step of quantum walk. The efficiency of the two algorithms is compared in Figure \ref{fig:Rastrigin_comparison_BBW_QW}, where the probabilities of successful samplings with respect to (w.r.t.) the number of iterations and the number of functional evaluations are compared with different values of rotation threshold $R_0$. At the initial stage of search with few iterations, quantum walk provides higher probabilities of success. It is seen that when $R_0=2$, the difference between the BBW-QW and BBW algorithms is the most significant. The benefit of quantum walk is also seen at the later stage of the search. The number of functional evaluations is a better criterion to evaluate the efficiency. Figure \ref{fig:Rastrigin_comparison_BBW_QW}-(b) illustrates the difference between the BBW and BBW-QW algorithms. Figure \ref{fig:Rastrigin_comparison_BBW_QW}-(c) compares the efficiencies of BBW and BBW-QW algorithm when the domain size is increased from $x\in [-5,5]$ to $x\in [-15,15]$. The efficiency of the BBW algorithms slightly decreases at the early search stage as the domain size increases, whereas it does not change much for the BBW-QW algorithm. It is also seen that with about 50 evaluations, both BBW and BBW-QW algorithms increase the probability of success to about 90\%. The difference between the two starts to emerge when more iterations are taken. 

To provide an overall picture of how the quantum search algorithms are compared with traditional global optimization methods, the probabilities of successful search w.r.t. the number of functional evaluations in the BBW, BBW-QW, simulated annealing, and genetic algorithms (GA) for Rastrigin function are compared in Figure \ref{fig:Rastrigin_comparison_BBW_GA}. The results from the genetic algorithms with different population sizes (5, 25, and 50) and simulated annealing with different initial temperature (100 and 1000) are shown. The optimum solution is known at $x=0$. When the distance between a located solution and the known optimum solution is less than a threshold value of $1.0\times 10^{-4}$, the search is regarded as a success. The threshold is chosen to be compatible with the resolution used in the quantum algorithms as a result of the number of available qubits. The number of iterations affect the probability of success. Among the three population sizes, the population size of $25$ is the best. Yet it is still much less efficient than the quantum search algorithms. Similarly, simulated annealing is not as efficient as the quantum search algorithms.

\begin{figure}
\subfigure[average PDF by BBW-QW algorithm ($R_0=2$)] {
  \includegraphics[width=0.5\textwidth]{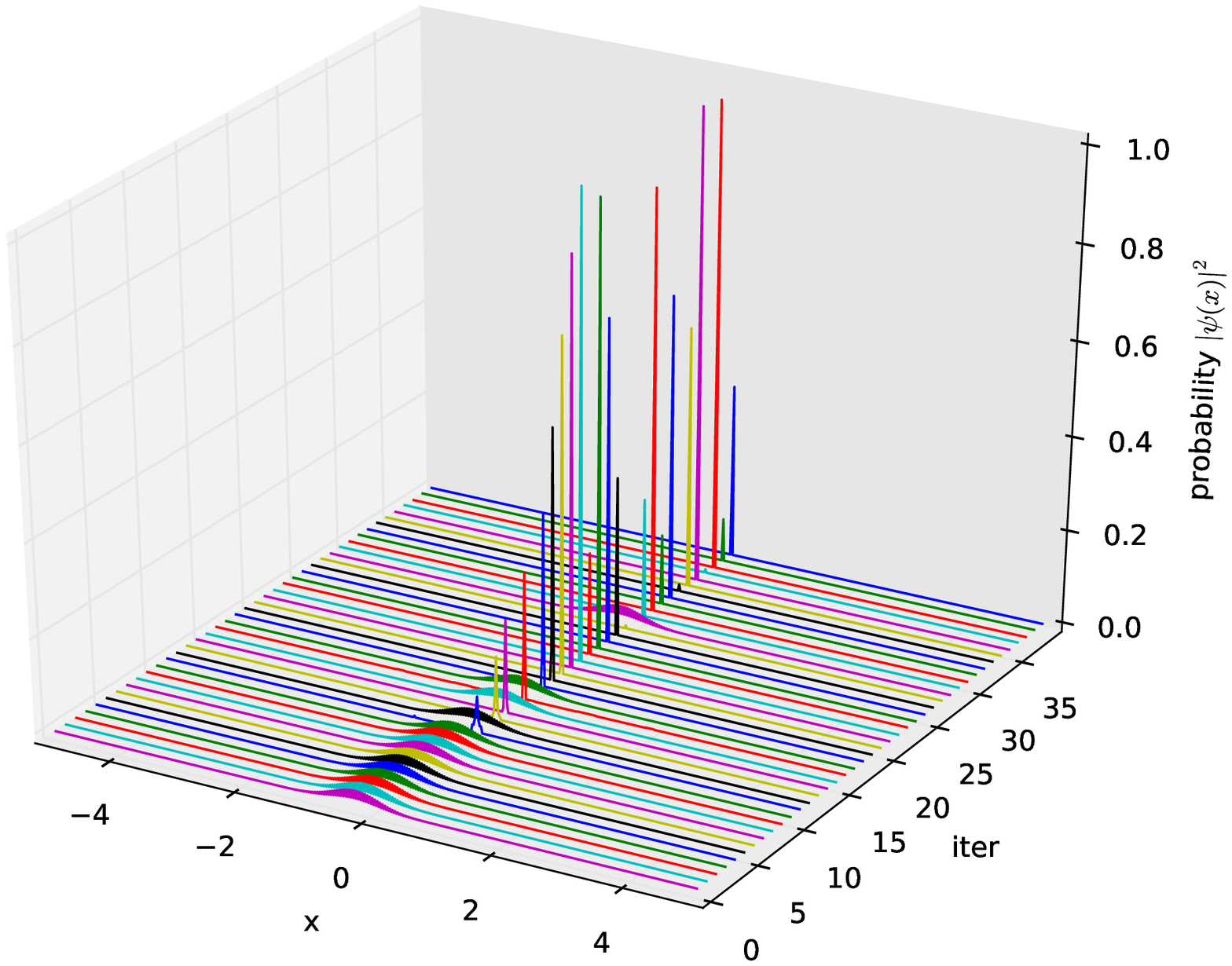} }
\subfigure[average PDF by BBW algorithm] {
  \includegraphics[width=0.5\textwidth]{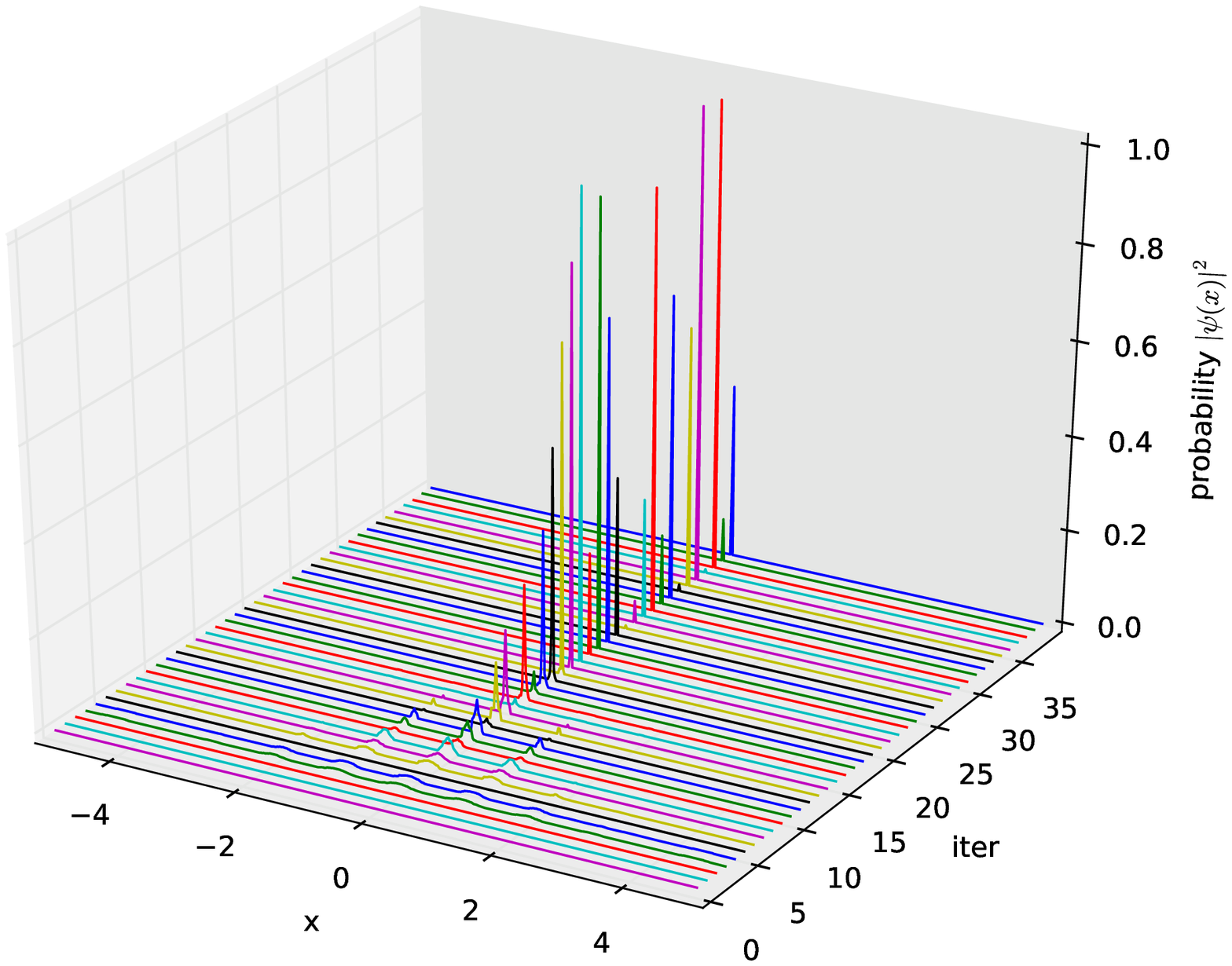} }
\caption{Comparison between average PDF's for Rastrigin function}
\label{fig:Rastrigin_avg}   
\end{figure}

\begin{figure}
\subfigure[typical PDF by BBW-QW algorithm ($R_0=2$)] {
  \includegraphics[width=0.5\textwidth]{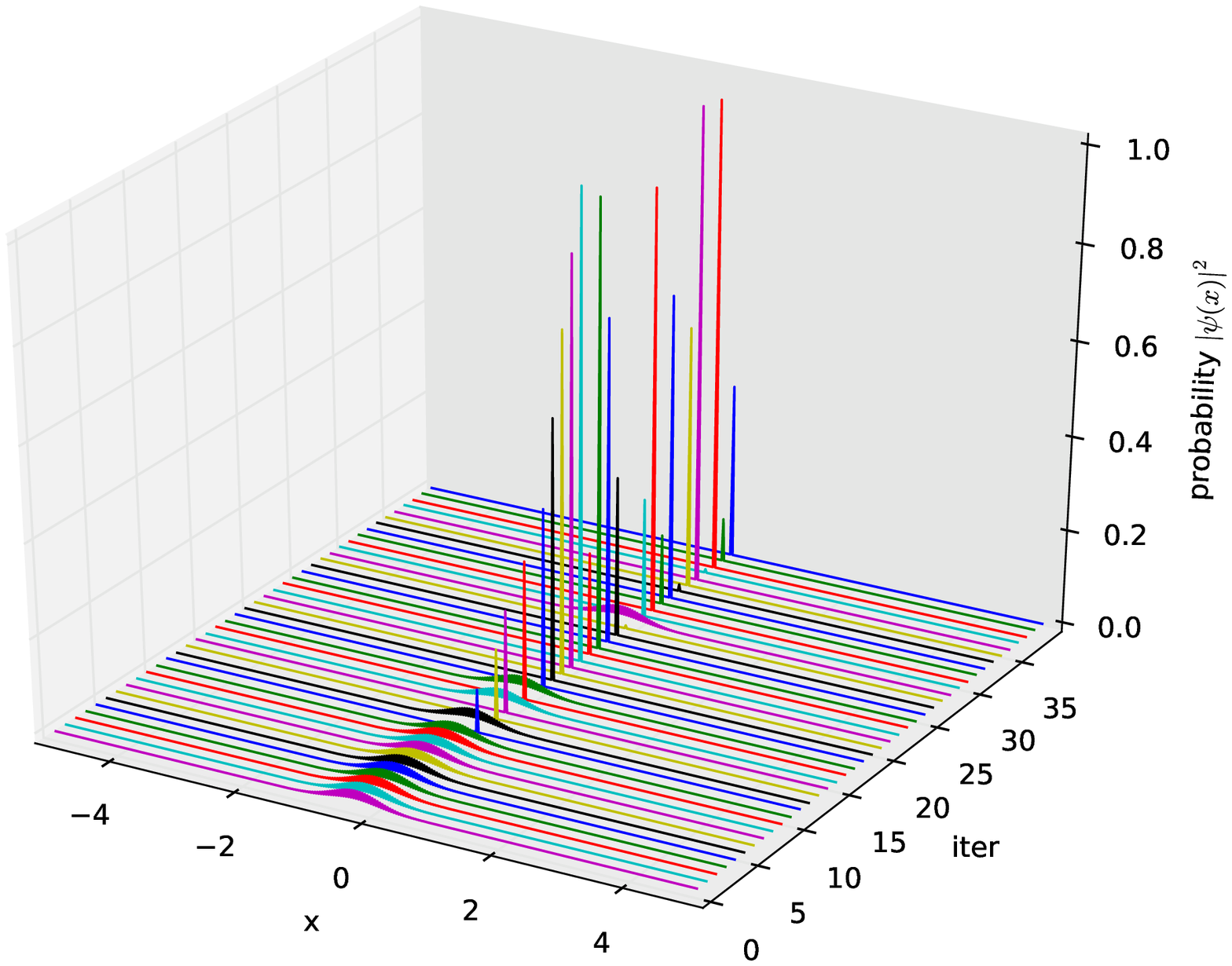} }
\subfigure[typical PDF by BBW algorithm] {
  \includegraphics[width=0.5\textwidth]{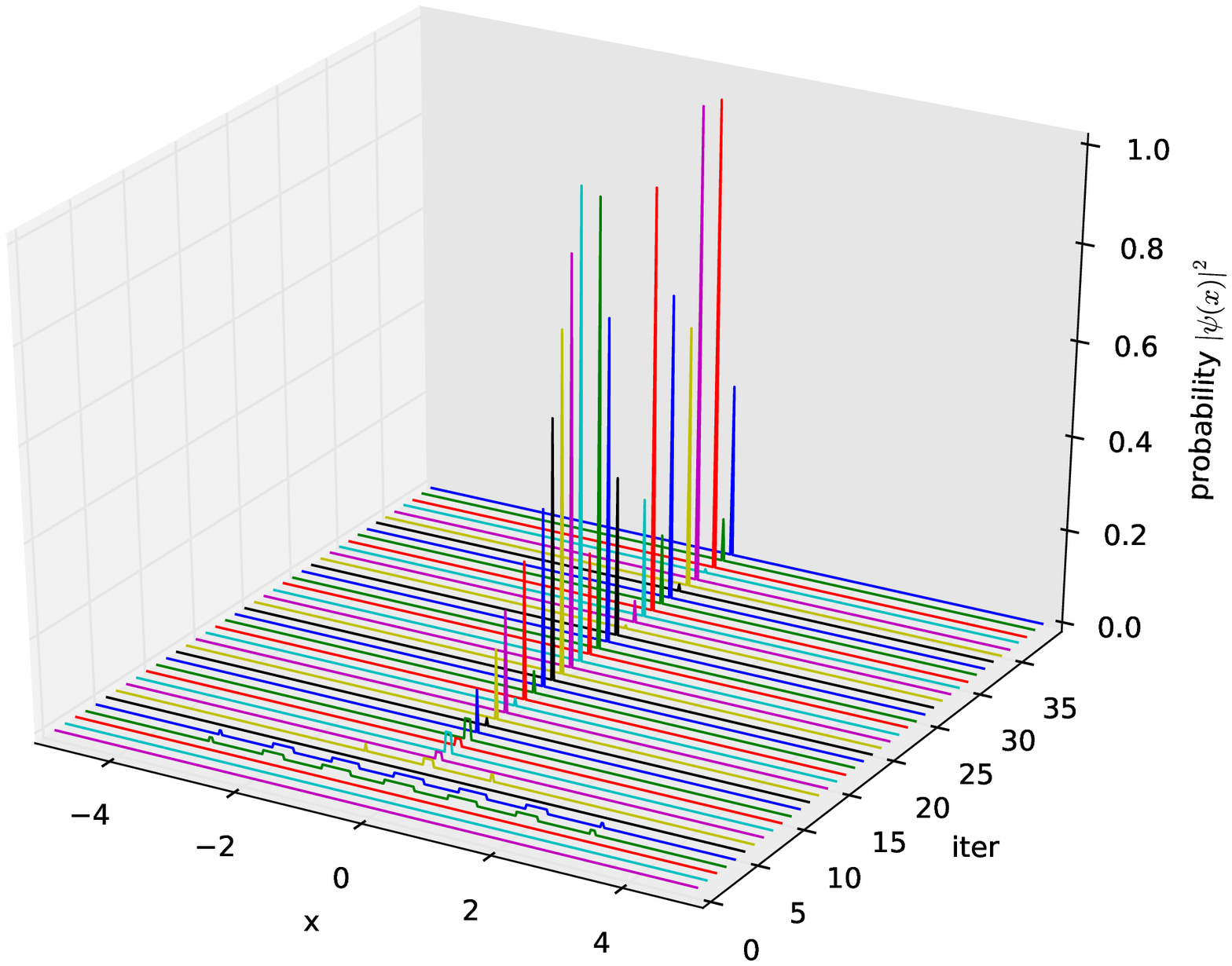} }
\caption{Comparison between typical PDF's for Rastrigin function}
\label{fig:Rastrigin}   
\end{figure}

\begin{figure}
\subfigure[probability of success w.r.t. iterations] {  \includegraphics[width=0.5\textwidth]{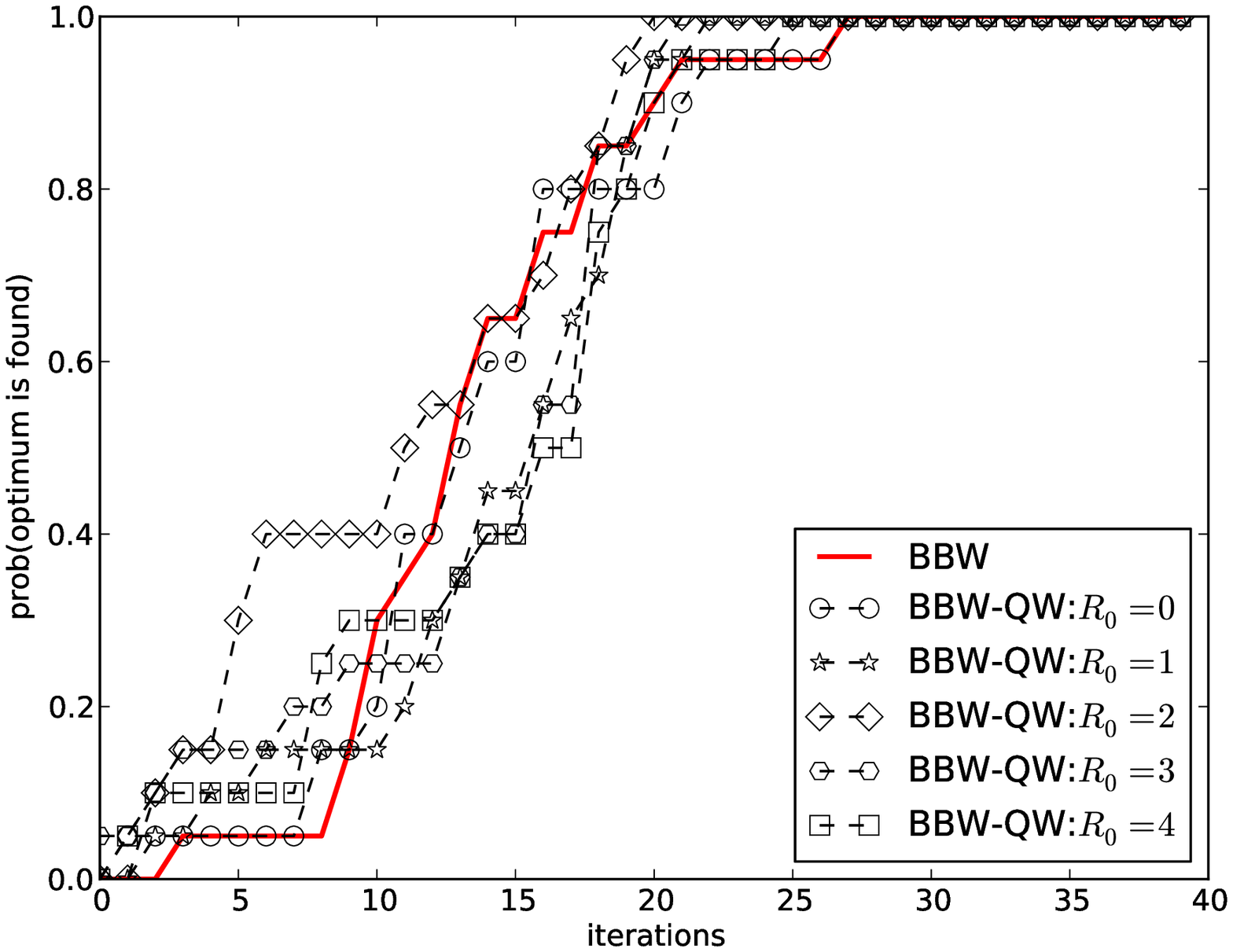} }
\subfigure[probability of success w.r.t. functional evaluations] { \includegraphics[width=0.5\textwidth]{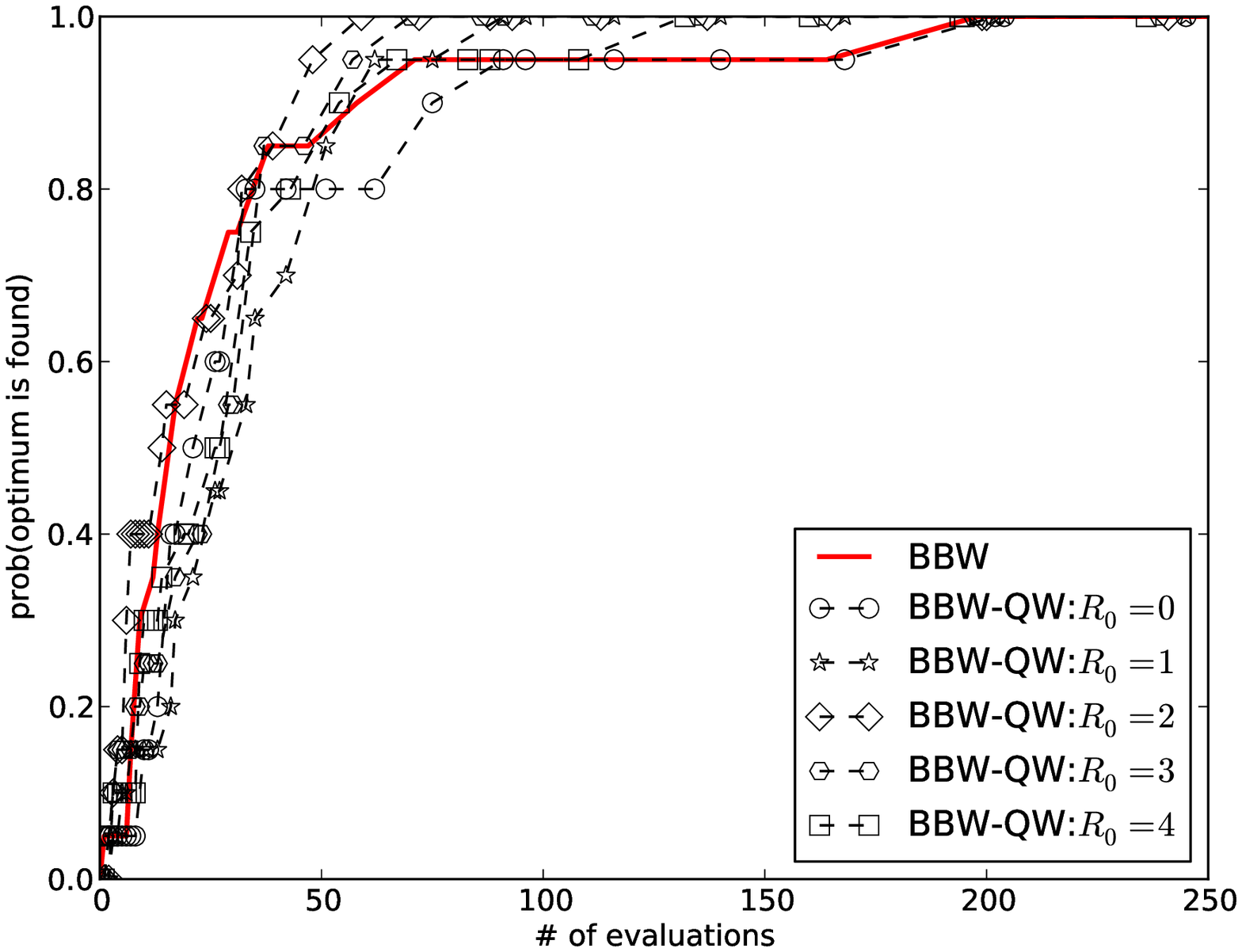} }
\subfigure[the effect of domain size w.r.t. functional evaluations] { \includegraphics[width=0.5\textwidth]{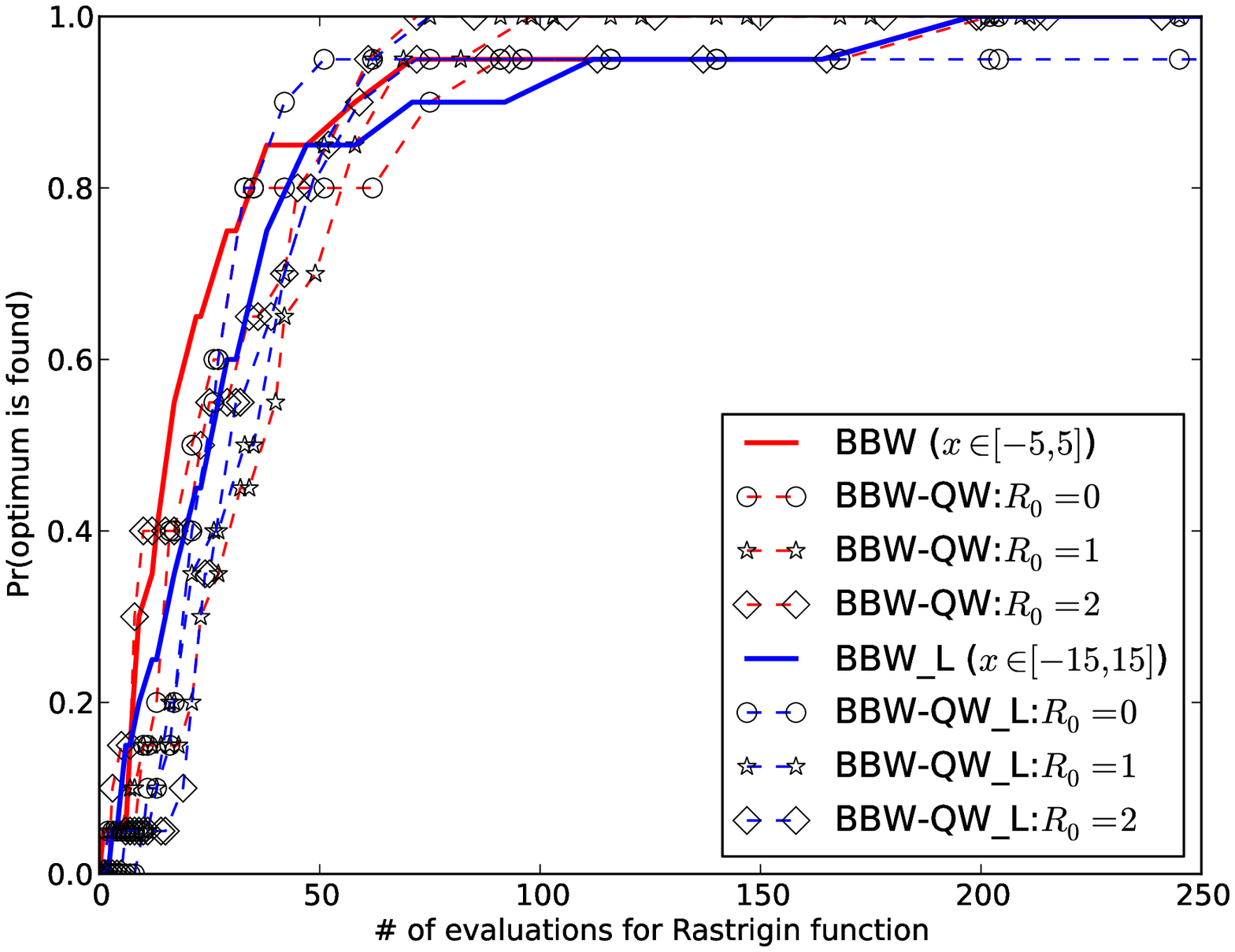} }
\caption{Comparison the efficiency of the BBW and proposed BBW-QW algorithms}
\label{fig:Rastrigin_comparison_BBW_QW}   
\end{figure}

\begin{figure}
\includegraphics[width=0.5\textwidth]{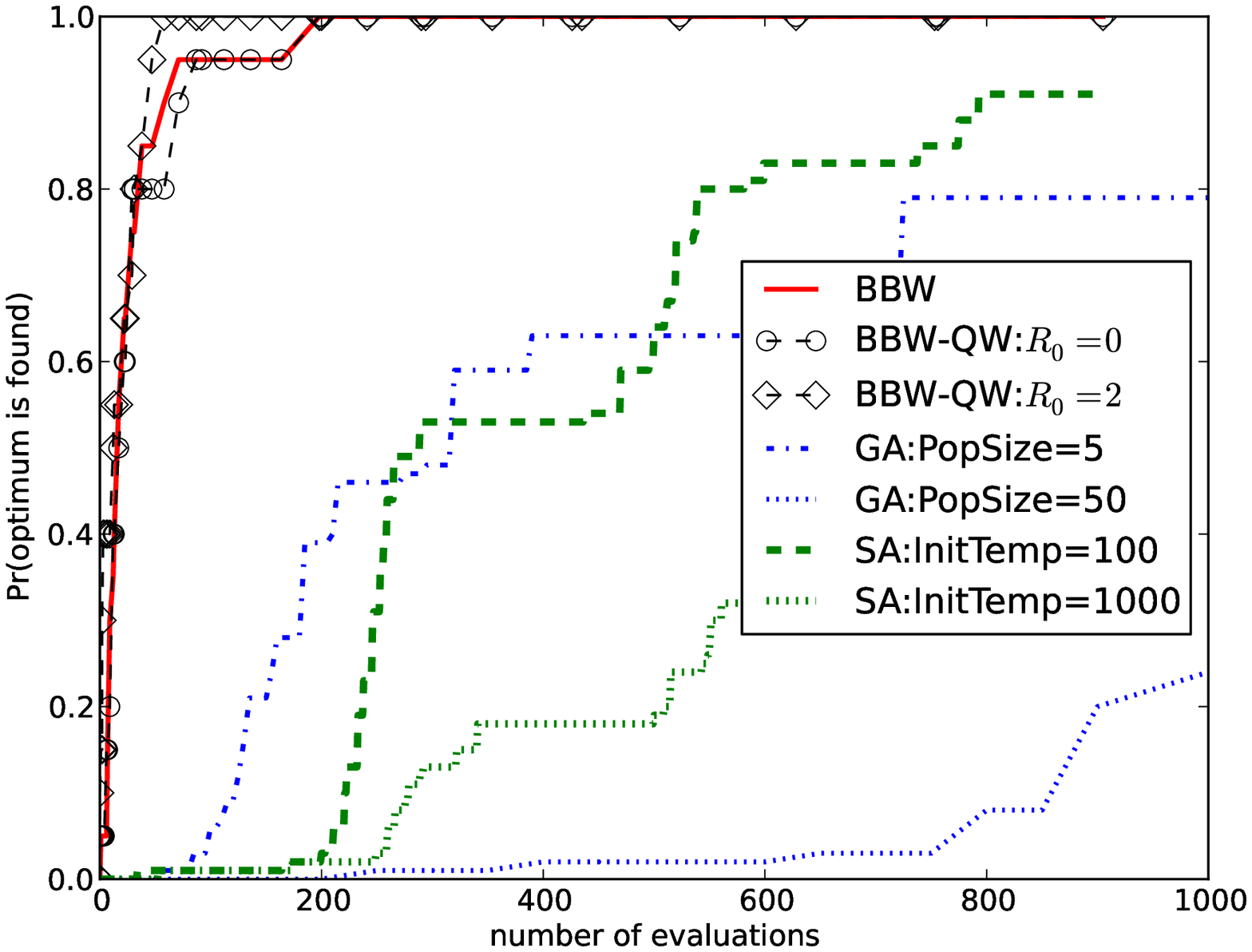} 
\caption{Comparison the efficiency of the BBW, GA, and simulated annealing algorithms for Rastrigin  function}
\label{fig:Rastrigin_comparison_BBW_GA}   
\end{figure}

The results of the BBW and BBW-QW algorithms in searching for Schwefel function are compared in Figure \ref{fig:Schwefel_avg} where $R_0=0$. The efficiencies of the BBW and BBW-QW algorithms with different $R_0$ values are also compared in Figure \ref{fig:Schwefel_comparison_BBW_QW}. It is seen that the quantum search algorithms work more efficiently for Schwefel function than for Rastrigin function. The optimum solution can be found with the probability of one with only few iterations. As a result, the difference between the two algorithms is relatively small.  Figure \ref{fig:Ackley_comparison_BBW_QW} compares the efficiencies of the BBW and BBW-QW algorithms for Ackley function. Similar to Rastrigin function, $R_0=2$ provides an obvious improvement for Schwefel and Ackley functions. It should be noted that the rotational threshold $R_0$ plays a key role of efficiency for the BBW-QW compared to BBW. If $R_0$ is too large, more quantum walks (with additional functional evaluations) are applied during the search, which will decrease the efficiency of the search algorithm. The test results show that a threshold value of $R_0<2$ is good for the test functions. In general, the selection of the value of $R_0$ depends on the complexity of the objective function. If the function has more local optima or wells in the search domain, more quantum walks are necessary, therefore a larger value of $R_0$ needs to be chosen.

\begin{figure}
\subfigure[average PDF by BBW-QW algorithm ($R_0=0$)] {
  \includegraphics[width=0.5\textwidth]{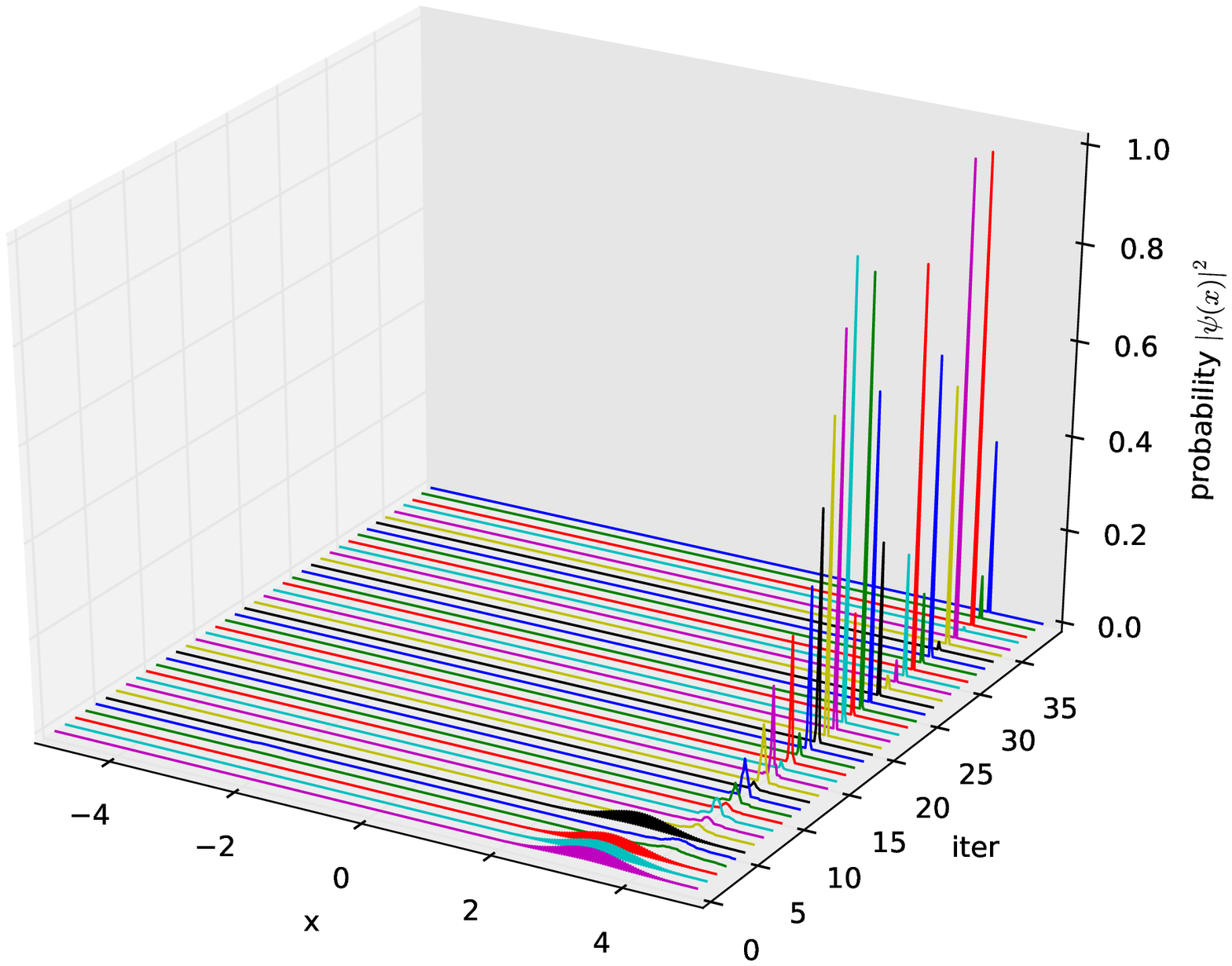} }
\subfigure[average PDF by BBW algorithm] {
  \includegraphics[width=0.5\textwidth]{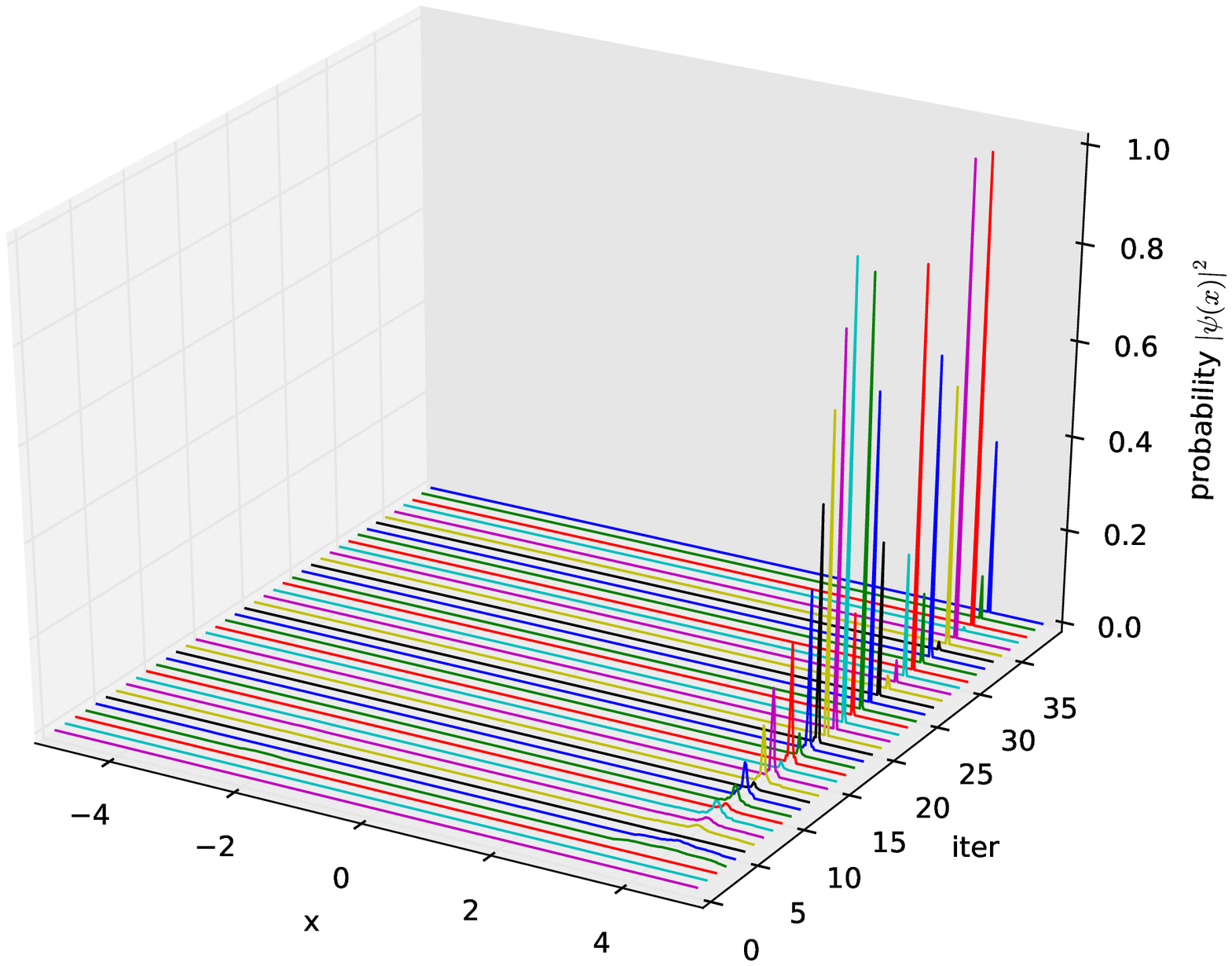} }
\caption{Comparison between average PDF's for Schwefel function}
\label{fig:Schwefel_avg}   
\end{figure}

\begin{figure}
\subfigure[probability of success w.r.t. iterations] {  \includegraphics[width=0.5\textwidth]{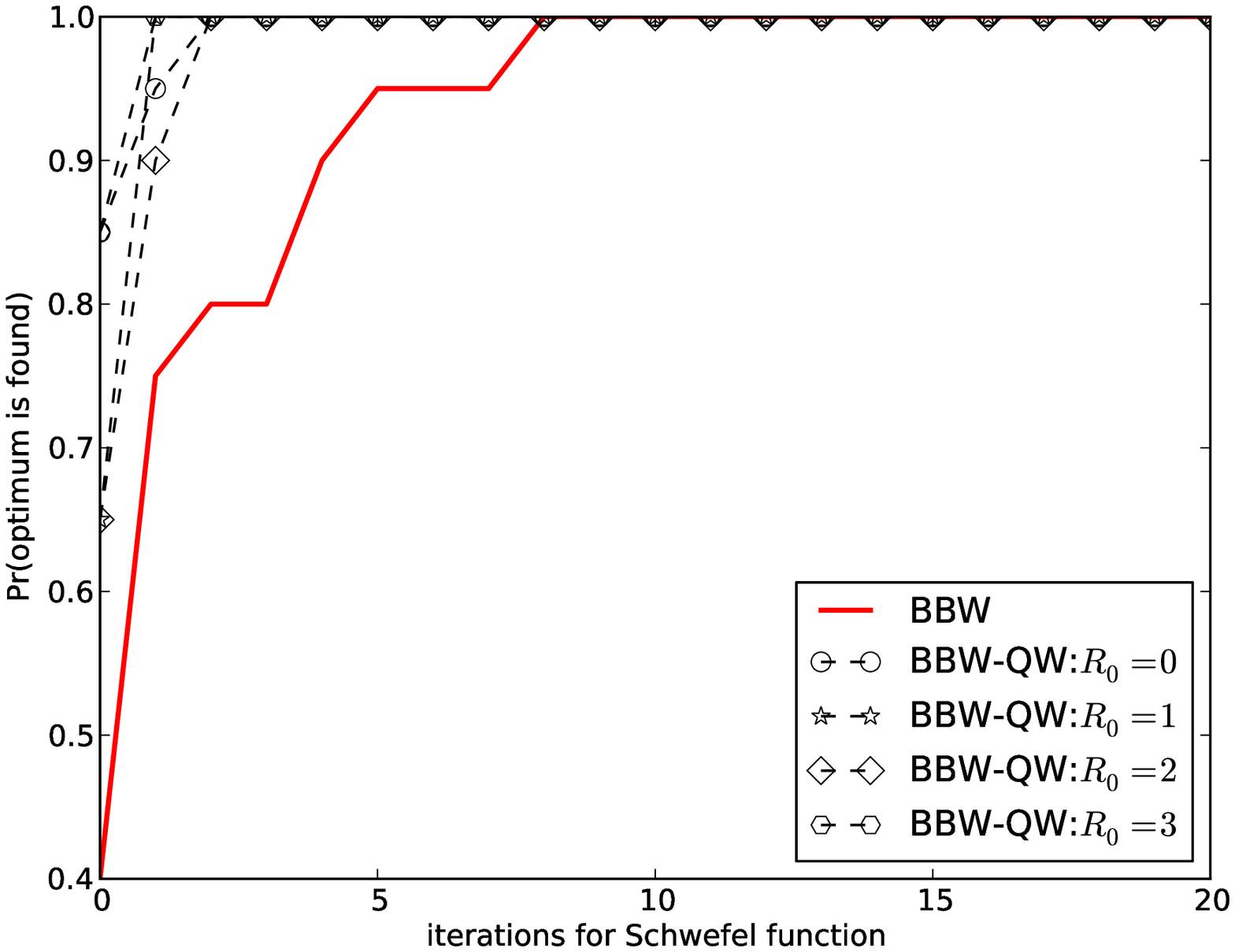} }
\subfigure[probability of success w.r.t. functional evaluations] { \includegraphics[width=0.5\textwidth]{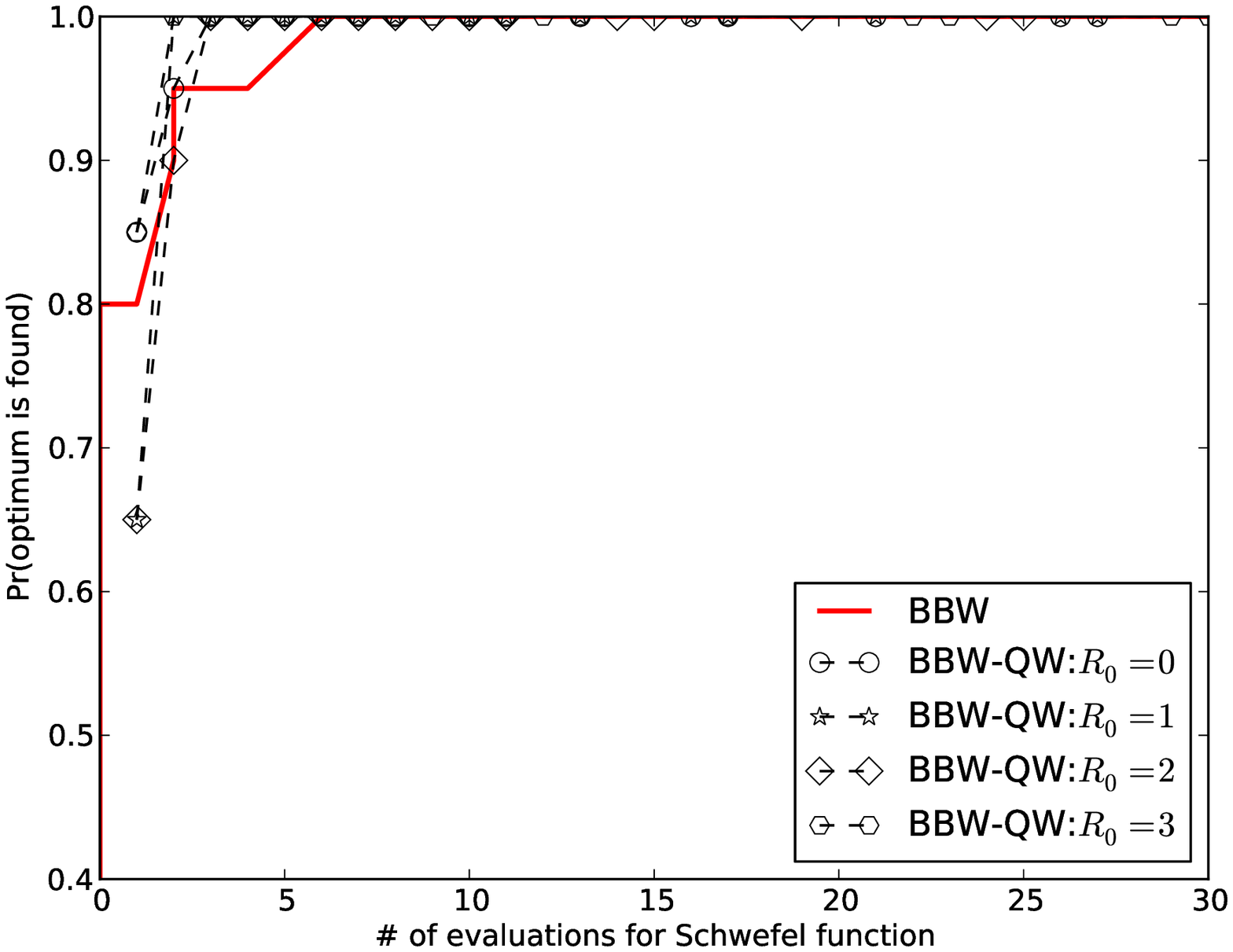} }
\caption{Comparison the efficiency of the BBW and proposed BBW-QW algorithms for Schwefel function}
\label{fig:Schwefel_comparison_BBW_QW}   
\end{figure}

\begin{figure}
\includegraphics[width=0.5\textwidth]{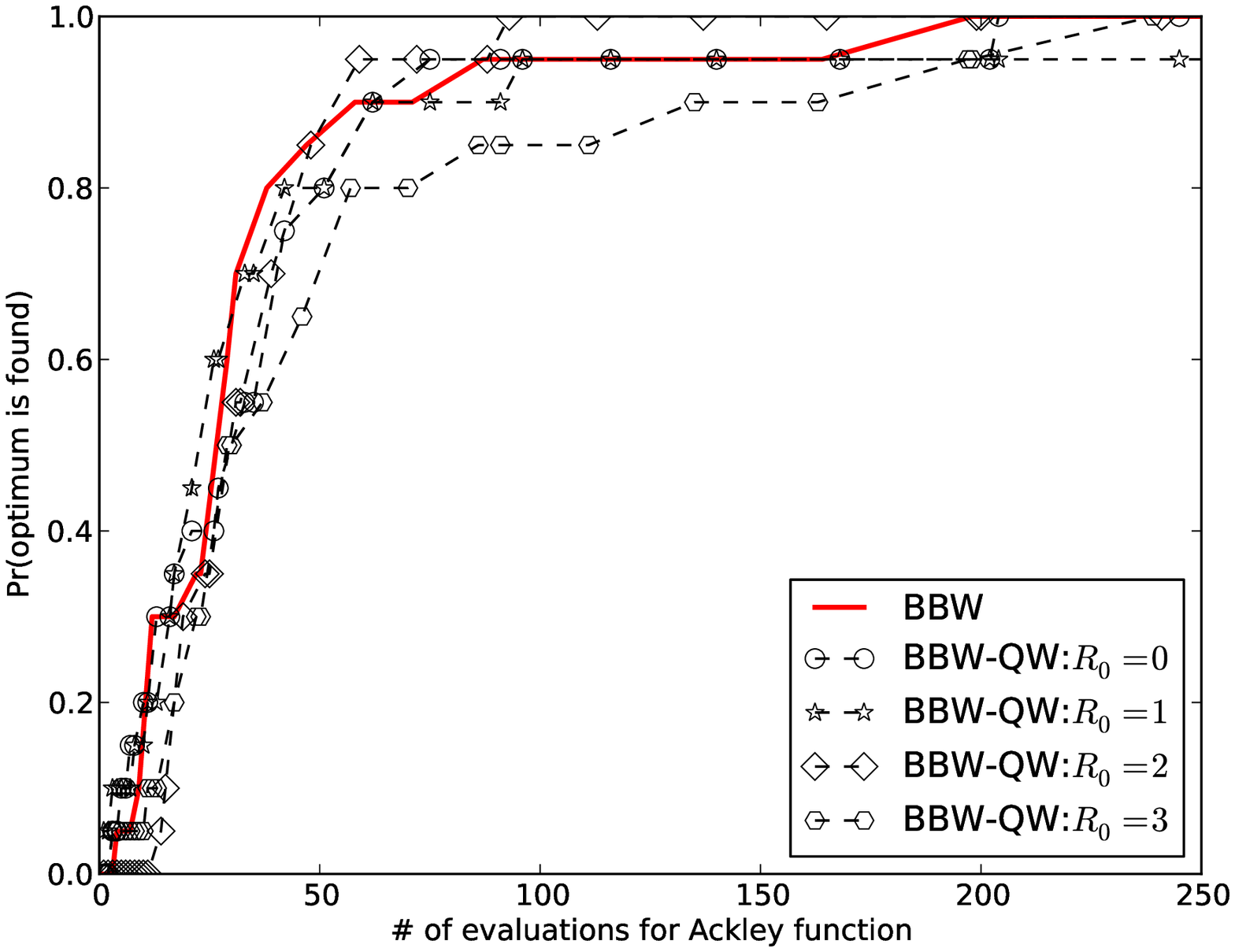} 
\caption{Comparison the efficiency of the BBW and proposed BBW-QW algorithms for Ackley function}
\label{fig:Ackley_comparison_BBW_QW}   
\end{figure}

\section{Concluding remarks}

In this paper, a hybrid approach that combines quantum walks with Grover search to solve global optimization problems is proposed. By taking advantages of quantum tunneling effect, quantum walks can enhance the traditional Grover search algorithm and improve the efficiency of search. The acceleration is achieved by quickly improving the threshold value at the early stage of search so that the solution space can be reduced faster during the Grover search.

Different from existing Grover search algorithms that focus on optimizing the number of Grover rotations only, the new algorithm tries to improve the search efficiency by accelerating the convergence of threshold value  toward the optimum. Nevertheless, as the threshold approaches the optimum value, the number of Grover rotations also increases. Therefore, a balance between the number of rotations and the number of iterations is needed for particular problems or applications. In an actual quantum computation environment, each sampling or measurement after performing a number of Grover rotations will actually destroy the quantum coherence. The amplitudes of the system will turn into one for the measured solution and zeros for all others. For each iteration, the Grover rotation always starts from the uniform distribution. Therefore, there is an overhead when the quantum register is initialized by Hadamard operation for each iteration. Reducing the number of iterations thus can improve the efficiency of computation in general.

\bibliographystyle{siamplain}      % mathematics and physical sciences

%\bibliography{QWOpt_bib}   % name your BibTeX data base

\end{document}